\documentclass[sigconf]{aamas} 

\usepackage{balance} % for balancing columns on the final page
\usepackage{amsfonts}
\usepackage{color}
\usepackage{graphicx}
\usepackage{verbatim}
\usepackage{fancyhdr}
\usepackage{amsthm}
\usepackage{stackrel}
\usepackage{caption}
\usepackage{subcaption}
\usepackage{amsmath} % For math environments
\usepackage{xparse}  % For defining commands
\usepackage{bbm}
\usepackage{multicol}
\usepackage{xcolor}

\setcopyright{ifaamas}
\acmConference[OptLearnMAS-24]{Appears at the 15th Workshop on Optimization and Learning in Multiagent Systems Workshop (OptLearnMAS 2024). Held as part of the Workshops at the 23rd International Conference on Autonomous Agents and Multiagent Systems.}{May 6 -- 10, 2024}{Auckland, New Zealand}{N.~Alechina, V.~Dignum, M.~Dastani, J.S.~Sichman (eds.)}

\copyrightyear{2024}
\acmYear{2024}
\acmDOI{}
\acmPrice{}
\acmISBN{}

\acmSubmissionID{<<EasyChair submission id>>}

\title[Satisfaction and Regret in Stackelberg Games]{Satisfaction and Regret in Stackelberg Games}
\author{Langford B White}
\affiliation{
  \institution{The University of Adelaide}
  \city{Adelaide}
  \country{Australia}}
\email{lang.white@adelaide.edu.au}

\author{Duong D Nguyen}
\affiliation{
  \institution{The University of Adelaide}
  \city{Adelaide}
  \country{Australia}}
\email{duong.nguyen@adelaide.edu.au}

\author{Hung X Nguyen}
\affiliation{
  \institution{The University of Adelaide}
  \city{Adelaide}
  \country{Australia}}
\email{hung.nguyen@adelaide.edu.au}

\newcommand{\BibTeX}{\rm B\kern-.05em{\sc i\kern-.025em b}\kern-.08em\TeX}

\newcommand{\Bc}{\ensuremath{\mathcal{B}}}

\newcommand{\Sc}{\ensuremath{\mathcal{S}}}

\renewcommand{\Pr}{\ensuremath{\mathbf{P}}}

\newcommand{\Ibb}{\ensuremath{\mathbb{I}}}

\newcommand{\Er}{\mathbf{E}}

\newcommand{\lb}{\ensuremath{\left(}}
\newcommand{\rb}{\ensuremath{\right)}}
\newcommand{\lbr}{\ensuremath{\left\{}}
\newcommand{\rbr}{\ensuremath{\right\}}}

\newcommand{\ra}{\ensuremath{\rightarrow}}

\newcommand{\bcr}{\begin{color}{red}}
\newcommand{\bcb}{\begin{color}{blue}}
\newcommand{\bcg}{\begin{color}{green}}
\newcommand{\ec}{\end{color}}

\newtheorem{theorem}{Theorem}

\newtheorem{proposition}{Proposition}

\begin{abstract}
This paper introduces the new concept of (follower) {\it satisfaction} in Stackelberg games and compares the standard Stackelberg game with its satisfaction version. Simulation results are presented which suggest that the follower adopting satisfaction generally increases leader utility. This important new result is proven for the case where leader strategies to commit to are restricted to be deterministic (pure strategies). The paper then addresses the application of regret based algorithms to the Stackelberg problem. Although it is known that the follower adopts a no-regret position in a Stackelberg solution, this is not generally the case for the leader. The report examines the convergence behaviour of unconditional and conditional regret matching (RM) algorithms in the Stackelberg setting. The paper shows that, in the examples considered, that these algorithms either converge to any pure Nash equilibria for the simultaneous move game, or to some mixed strategies which do not have the ``no-regret'' property. In one case, convergence of the conditional RM algorithm over both players to a solution ``close'' to the Stackelberg case was observed. The paper argues that further research in this area, in particular when applied in the satisfaction setting could be fruitful. \\

\end{abstract}


\begin{document}

\pagestyle{fancy}
\fancyhead{}

\maketitle

%------------------------------------------------------------------------------
\section{Introduction}

This paper considers a set of game-theoretic models in which there is a player that is able to benefit from moving first. Such a player may {\it commit} to a strategy before the game is played. 
The two-player case is a well-studied problem in economics due to von Stackelberg \cite{Stack}, and is often known as a {\it Stackelberg duopoly}. The Stackelberg duopoly is a model of two firms competing on quantity (Cournot competition) where one firm is dominant (i.e. the ``market leader''). The dominant player moves first and is called {\it the leader} in game theory. The other player is called {\it the follower}.  Mixed (randomised) strategies are permitted for the leader and generally advantageous. Since it is maximising a linear function, the follower always plays a pure (deterministic) strategy (but this may not be unique). The model assumes that the follower observes the leader's action, and plays best response (BR). The leader knows that this is the follower's strategy. The leader knows the complete game structure, whereas the follower only need to know its own utility function (and associated action spaces). \footnote{More precisely, the leader only requires knowledge of its own utilities and the BR function for the follower. This is important when the leader is trying to learn the follower behaviour in order to construct its optimal policy to commit to.}\\

Stackelberg game models provides valuable insights into strategic decision-making, equilibrium analysis, and policy design in complex systems. Thus, the Stackelberg framework finds applications in various real-world domains~\cite{Harr24}. Examples include pricing strategies, cybersecurity, traffic management, and healthcare. We are interested in Stackelberg game models in wargaming (e.g. \cite{LBW_AIWG_1}) because we can think of there being two coalitions (the players) who take turns in a strict ordering. The familiar ``red team - blue team'' scenario is an example, where one team moves first. Indeed, it's possible to think of the Stackelberg duopoly as situation where the leader can commit to its strategy knowing that the follower must play BR. So, for example, if the blue team is dominant, then it can ``design'' its (mixed) strategy knowing that it can attain a certain utility irrespective of what the red team does. Also, the there is no incentive for the red team to deviate from BR - it will suffer a loss in its utility by doing so. This kind of idea has been applied to the problem of protecting infrastructure in \cite{Brown,Pita09b}, for example.\\

The focus of this paper is the idea of {\it satisfaction} in Stackelberg games. Satisfaction games arise when rather than seeking to maximise their utilities, agents seek to achieve a certain minimum level of utility (e.g. \cite{GPYA17}). Such an agent is then said to be satisfied. Indeed, the concept of satisfaction can be generalised to encode a number of set theoretic constraints which specify ``satisfaction'' when met \cite{GPYA17}. In \cite{WRN20}, the authors have developed the concept of  {\it Bayesian rationality} in satisfaction games. In this setting, agents seek to maximise their probability of satisfaction, even if they are unable to attain satisfaction with probability one. In the context of a Stackelberg duopoly, we can consider either the follower, or the leader (or both) as seeking satisfaction and we'll see there is a natural interpretation in terms of probability of satisfaction which relates to the Bayesian concept of information in strategic games (e.g. \cite{Aum87}).\\

The outline of the paper is as follows. In sec. \ref{sec:Stack_Sat}, we introduce the new idea of Bayesian satisfaction in a Stackelberg duopoly. In this paper, we only consider satisfaction as expressed in terms of achieving a certain specified minimum utility. However, we point out the much more general, and potentially more widely applicable formulation in terms of {\it correspondence functions} \cite{GPYA17}. Section \ref{sec:Stack_Sat} revisits the linear programming (LP) solutions \cite{Con06,Con11} for standard Stackelberg games in the context of satisfaction. We prove a result which allows ``pruning'' of the set of possible follower best-responses based on a satisfaction domination concept. This result also yields an initial basic feasible solution for the multiple LP case, which might help improve the efficiency of LP solution methodologies when specialised to the satisfaction model. Simulations examples show that it appears to be the case that when the follower adopts satisfaction as its reward setting, optimal leader utility is higher than for the standard Stackelberg case. In the case where leader strategies are restricted to be pure actions (generally sub-optimal \cite{Con06}), we prove that this is indeed the case. This appears to be a significant new finding with important ramifications for game design more generally. In sec. \ref{sec:regret}, we consider the applicability of regret based adaptive procedures (e.g. \cite{Hart00,Hart05}) to the Stackelberg model. We know from \cite{Con11}, that the Stackelberg solution has the ``no-regret'' property for the follower, but not necessarily for the leader. A very recent study~\cite{Harr24} suggests that ``no-regret" is impossible for the leader when both follower and the context (of which the leader commits to a strategy based on it) are driven by an adversary, but possible when either one of them is chosen stochastically. This somewhat aligns with our finding which shows that the leader's utility can be higher than the standard Stackelberg case when the follower adopts satisfaction as its reward setting (instead of playing BR strategy). The authors in~\cite{Harr24}, however, study the online repeated game setting in which extra side information is availalbe to both players, while our work examines a one-off game and do not assume any additional formation. In addition, we focus on the new concept of (follower) satisfaction in Stackelberg games and on the application of regret-based algorithms to find the optimal strategy for the Stackelberg problem. Simulations are used to show that the regret-based methods considered here generally converge to pure Nash equilibria if they exist, or may not converge at all ! Interestingly, one example showed convergence of the conditional regret matching scheme of \cite{Hart00} ``near to'' the Stackelberg solution. Further study of these approaches is required. 
\section{Stackelberg Satisfaction Games\label{sec:Stack_Sat}}

We assume here that the reader has a basic understanding of basic game theory concepts and notation, as well as Stackelberg games in the standard (i.e. utility maximisation) setting (see e.g. \cite{LBW_AIWG_1}). In this section, we shall introduce the idea of a Stackelberg satisfaction game (SSG). In the first instance, we'll address the case where the follower agent seeks satisfaction whereas the leader still seeks to compute a strategy to commit to which maximises its utility. We'll consider here the case where satisfaction for the follower is expressed in terms of achieving a minimum value of utility $U^f_-$. Suppose the leader plays a mixed strategy $\pi_\ell$, then the follower best-response function $\Bc : \Delta (\Sc^\ell ) \ra 2^{\Sc^f}$, is given by
\begin{align}
\Bc(\pi_\ell) & = \ \stackrel[t \in \Sc^f]{}{\text{argmax}} \ \sum_{s \in \Sc^\ell} \ \pi_\ell(s) \ \Ibb \lbr U^f \lb s, t \rb \geq U^f_- \rbr \nonumber \\
& =  \ \stackrel[t \in \Sc^f]{}{\text{argmax}} \ \Pr_{\pi_\ell} \lbr U^f(.,t) \geq U^f_- \rbr \ ,
\label{eq:BR1}
\end{align} 
where $\Ibb\{.\}$ denotes the indicator function. Note that $\Bc$ is generally set-valued. Observe that the summation on the r.h.s. of \eqref{eq:BR1} is the probability (under $\pi_\ell$) that the follower is satisfied when it chooses action $t$. Thus the best-response to a (mixed) leader strategy $\pi_\ell$ is one which maximises the follower's probability of satisfaction. \\

The leader then chooses the mixed strategy
\begin{align}
\pi_\ell^* & = \ \stackrel[\pi \in \Delta(\Sc^\ell)]{}{\text{argmax}} \ \sum_{s \in \Sc^\ell} \ \pi(s) \ \stackrel[t \in \Bc(\pi)]{}{\text{max}} \ U^\ell(s,t) \nonumber  \\
& =  \ \stackrel[\pi \in \Delta(\Sc^\ell)]{}{\text{argmax}} \ \Er_\pi \lbr \stackrel[t \in \Bc(\pi)]{}{\text{max}} \ U^\ell(.,t)  \rbr \ .
\label{eq:lead_act1}
\end{align}
This is the optimal mixed strategy to commit to. Ties are resolved in favour of the leader in \eqref{eq:lead_act1}.\\

{\it Comment : }
%\vspace{-3mm} 
We could consider satisfaction from the point of view of the leader, but the presence of ties in the BR function for the follower would appear to prevent a simple interpretation in terms of probability of leader satisfaction. 

In the sequel, we'll consider only the case for follower satisfaction. 

\subsection{Solution via Multiple Linear Programs\label{ssec:mult_LP}}

In \cite{Con06}, a solution strategy for the standard Stackelberg game using multiple linear programs (LP) is proposed. The idea is to construct a LP for each follower action where the constraints in each LP enforces the BR property for that follower action. So, in the standard setting, we have the LP for each follower action $s^f \in \Sc^f$, 
\begin{equation*}
\stackrel[\pi \in \Delta(\Sc^\ell)]{}{\text{max}} \ \Er_\pi \lbr U^\ell(.,s^f) \rbr %\nonumber \\
\end{equation*}
\begin{align}
\text{subject\ to\ :}  \ \ \forall \, t \in \Sc^f, \ \sum_{s \in \Sc^\ell} \pi(s) \ \lb U^f(s,s^f) - U^f(s,t) \rb \geq 0 \ ,
\label{eq:LP1}
\end{align}
together with the usual constraints on $\pi$ being a pmf. The objective function is the leader utility given the follower plays $s^f$, whilst the constraints enforce the property that $s^f$ is a BR to leader mixed action $\pi$. As pointed out in \cite{Con06}, not all these LPs are feasible, but at least one will be. Choose that feasible LP with maximum objective function (and solution $\pi_\ell^*$) indexed by $s^f = s_*^f$, to yield the optimal strategy pair $(\pi_\ell^*, s_*^f)$. \\

In the context of follower satisfaction, the BR constraints in the LP indexed by $s^f$ become \footnote{We reverse the inequality to maintain consistency with linear programme implementations such as matlab's \texttt{linprog} function.}
\begin{align}
\forall \, t \in \Sc^f, \ \sum_{s \in \Sc^\ell} \pi(s) \ \lb \Ibb \lbr U^f(s,t) \geq U^f_- \rbr - \Ibb \lbr U^f(s,s^f) \geq U^f_- \rbr \rb \leq 0 \ .
\label{eq:LP2}
\end{align}
 The coefficients in the matrix specifying these constraints contains only the values $0, \pm 1$. More specifically, the coefficient (corresponding to the pair $(t,s)$) is zero if the follower is satisfied under both $(s,t)$ and $(s,s^f)$ or is not satisfied under either of these action pairs. The coefficient will be minus one if the follower is satisfied under $(s,s^f)$ but not under $(s,t)$. Conversely, the coefficient will be plus one  if the follower is satisfied under $(s,t)$ but not under $(s,s^f)$. Several observations can be made:
 
 \begin{itemize}
 \item If either $U_-^f \ra -\infty$ (i.e. the follower is always satisfied) or $U_-^f \ra +\infty$ (i.e. the follower is never satisfied), the optimal solution is the pure strategy pair $(s^\ell_*, s^f_*)$ maximising $U^\ell$ (if this maximum is unique). This occurs because, in each case, the follower BR function $\Bc(\pi) = \Sc^f$. All the constraint matrices for each LP in \eqref{eq:LP2}, are zero (the indicator functions are either both zero or both one), so the constraints are automatically met with equality. Thus the LP indexed by $s^f$ yields the pure leader action maximising $U^\ell(.,s^f)$. \\
 
 \item Due to the specific nature of the constraints \eqref{eq:LP2}, the LP solution (in terms of basic feasible solutions) can be found more easily. Indeed, let write \eqref{eq:LP2} as $A(s^f) \, \pi \leq 0$, where $A$ is $N^f \times N^\ell$ with elements
\begin{align*}
\left[ A(s^f) \right]_{t,s} & = \Ibb \lbr U^f(s,t) \geq U^f_- \rbr - \Ibb \lbr U^f(s,s^f) \geq U^f_- \rbr \ .
\end{align*}
Appending $N^f$ slack variables $y$ yields the standard form
\begin{equation*}
\max_{\pi} \ U^\ell(.,s^f)^T \ \pi 
\end{equation*}
\begin{align*}
\text{subject\ to:} \quad
\begin{cases}
A \, \pi + y = 0 \\
\mathbf{1}^T \, \pi = 1 \\
\pi,y \geq 0 \ .
\end{cases}
\end{align*}
Now, suppose there is a column $a_k$ of $A$ which contains no $+1$ entries. Then $\pi = e_k, y = -a_k$ is a basic feasible solution (bfs) to the LP (here $e_k$ denotes the unit vector with a one in entry $k$ and zeros elsewhere and $\mathbf{1}$ is the vector of all ones). In the case where no such column can be found, then every column of $A$ as at least one entry that is $+1$. Thus for every $s \in \Sc^\ell$, there is a $t \in \Sc^f$ with $A(t,s) = 1$. This can only occur if $U^f(s,s^f) < U^f_-$ for all $s \in \Sc^\ell$, i.e. the follower is never satisfied playing action $s_f$.  A consequence of this is that \eqref{eq:LP2} becomes
\begin{align*}
\forall \, t \in \Sc^f, \ \sum_{s \in \Sc^\ell} \pi(s) \ \Ibb \lbr U^f(s,t) \geq U^f_- \rbr  \leq 0 \ .
\end{align*} 
So for any $(s,t)$ with $U^f(s,t) \geq U^f_- $, (i.e. there is a +1 entry in column $s$ of $A$), then $\pi(s) = 0$. However, following this argument, every column of $A$ has a +1 entry so $\pi(s) = 0$ for all $s$, which violates the equality constraint. So this LP indexed by $s^f$ is infeasible and thus can be ignored. This is unsurprising since such a follower action $s^f$ is dominated by all others. We thus have the following theoretical results. 

\end{itemize}

\begin{theorem}\label{thrm_1}
If the constraint matrix appearing in the LP indexed by follower BR $s^f$ has at least one +1 entry in every column, then that LP is infeasible.
\end{theorem}

{\it Comment :}
%\vspace{-3mm}
The converse of the above result also holds, because, as described above, if there is a column of $A$ which has all its elements being non-positive, then a bfs can always be constructed. \\
 
Conceptually, it would appear that the leader could attain a higher utility in a satisfaction game because the best responses of the follower are less constrained. Indeed, we have :

\begin{proposition}
In a Stackelberg satisfaction game, the expected leader utility is never lower than that attained in the standard Stackelberg game with the same utilities.
\end{proposition}

We can prove a more restrictive version of Proposition 1 if we restrict leader strategies to be deterministic, i.e. the optimal {\it pure} strategy to commit to :

\begin{theorem}\label{thrm_2}
In a Stackelberg game restricted to deterministic (pure) strategies, the leader's utility in a satisfaction version is never lower than that for the standard utility-based game.
\end{theorem}

\begin{proof}
Consider the forms of the BR function for the follower for pure leader strategies in the standard and satisfaction models :
\begin{align*}
\Bc(s^\ell) & = \ \stackrel[s^f \in \Sc^f]{}{\text{argmax}} \ U^f \lb s^\ell, s^f \rb \ , \\
\Bc_\text{sat}(s^\ell ; U_-^f) & = \ \stackrel[s^f \in \Sc^f]{}{\text{argmax}} \ \mathbb{I} \lbr U^f \lb s^\ell, s^f \rb \geq U_-^f \rbr \ , 
\end{align*}
both of which may are generally set-valued. Define the following subsets of $\Sc^\ell \times \Sc^f$
\begin{align}
\Gamma_\text{max} & = \lbr (s^\ell, s^f) : s^f \in \Bc(s^\ell) \rbr\ , \nonumber \\
\Gamma_\text{sat}(U_-^f) & = \lbr (s^\ell,s^f) : s^f \in \Bc_\text{sat}(s^\ell ; U_-^f )\rbr \ .
\label{eq:B_sets}
\end{align}
When $U_-^f \ra - \infty$, then $\Gamma_\text{sat}(U_-^f) = \Sc^f$ because the indicator functions in \eqref{eq:B_sets} take the value one for any pair $(s^\ell,s^f)$. As $U_-^f$ is increased, for each $s^\ell$, the number of follower actions yielding satisfaction decreases, but as long as there is at least one such action giving satisfaction, the maximum remains one. For each given $s^\ell$, if $U_-^f$ sufficiently large, there will be no follower actions such that $(s^\ell,s^f)$ yields satisfaction, so all indicator functions take the value zero, so that again $\Gamma_\text{sat}(U_-^f) = \Sc^f$. Firstly assume that, given $U_-^f$, there is at least one pair of actions $(s^\ell,s^f)$ for which $U^f \lb s^\ell, s^f \rb \geq U_-^f $. Suppose $(s_*^\ell,s_*^f) \in \Gamma_\text{max}$, then it must then hold that $U^f \lb s_*^\ell, s_*^f \rb \geq U_-^f $, i.e. $(s_*^\ell,s_*^f) \in \Gamma_\text{sat}(U_-^f)$. In the case where there are no pair of actions $(s^\ell,s^f)$ for which $U^f \lb s^\ell, s^f \rb \geq U_-^f $, then $\Gamma_\text{sat}(U_-^f) = \Sc^f$. In each case, it therefore holds that $\Gamma_\text{max}\subseteq \Gamma_\text{sat}(U_-^f)$. The maximum leader utilities for the standard and satisfaction game are \footnote{Note that this formulation resolves ties in the BR functions in favour of the leader as required in the Stackelberg solutions.}
\begin{align*}
U_\text{max}^\ell & = \max \lbr U^\ell(s^\ell,s^f) : (s^\ell,s^f) \in \Gamma_\text{max} \rbr \\
U_\text{sat}^\ell(U_-^f) & = \max \lbr U^\ell(s^\ell,s^f) : (s^\ell,s^f) \in \Gamma_\text{sat}(U_-^f) \rbr \ ,
\end{align*}
so it follows that $U_\text{max}^\ell \leq U_\text{sat}^\ell(U_-^f)$. The above argument holds for any satisfaction threshold $U_-^f$, hence the result is proven.
\end{proof}

\subsection{Examples\label{ssec:examp}}

In this section, we use some numerical examples to illustrate the properties of the satisfaction formulation of the Stackelberg game. In particular, we'll consider the behaviour of the solutions as the follower utility bound is varied, and characterise the leader utility compared to the standard Stackelberg formulation. \\

In this example, there are $N^\ell$ leader actions, and $N^f$ follower actions. The leader and follower utilities were randomly generated for each round, with individual utilities i.i.d. uniform $(0,1)$ random variables. Thus the maximum utility attainable for the leader is a random variable $Y$ distributed according to $F_Y(y) = y^{N^f  N^\ell}$ supported on $(0,1)$. 

\begin{itemize}
    \item In the first simulation, we used $N^\ell = 10, N^f = 5$ with 40 uniformly spaced satisfaction levels. We averaged 1000 independent trials to obtain the results shown in fig. \ref{fig:sim1}. 
    \item In the second simulation, we used $N^\ell = 10, N^f = 10$ with 20 uniformly spaced satisfaction levels. We averaged 500 independent trials to obtain the results shown in fig. \ref{fig:sim2}. 
    \item In the third simulation, we used $N^\ell = 10, N^f = 15$ with 20 uniformly spaced satisfaction levels. We averaged 500 independent trials to obtain the results shown in fig. \ref{fig:sim3}. 
    \item In the fourth simulation, we used $N^\ell = 5, N^f = 5$ with 20 uniformly spaced satisfaction levels. We averaged 1000 independent trials to obtain the results shown in fig. \ref{fig:sim5}.
    \item In the fourth simulation, we used $N^\ell = 5, N^f = 10$ with 20 uniformly spaced satisfaction levels. We averaged 1000 independent trials to obtain the results shown in fig. \ref{fig:sim4}.
    \item In the fourth simulation, we used $N^\ell = 5, N^f = 15$ with 20 uniformly spaced satisfaction levels. We averaged 1000 independent trials to obtain the results shown in fig. \ref{fig:sim6}. 
\end{itemize}

\begin{figure*}[!t] 
        \centering
        \begin{subfigure}[b]{0.5\linewidth}
                \includegraphics[scale=0.30]{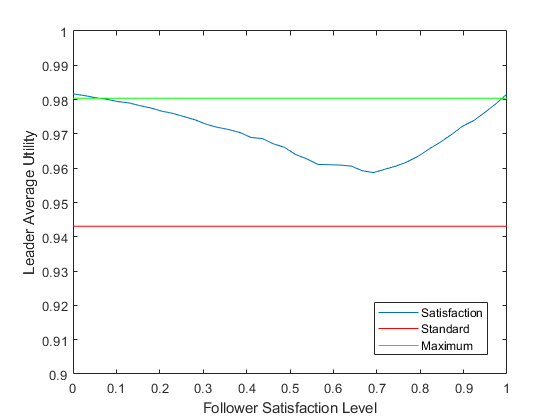}
                \includegraphics[scale=0.30]{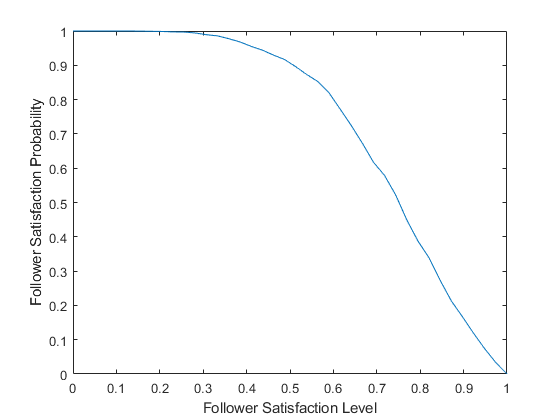}
                \caption{Simulation 1: The maximum possible leader utility of 0.980.\\ Here, $N^\ell = 10, N^f = 5$. \label{fig:sim1}}
        \end{subfigure}%
        \begin{subfigure}[b]{0.5\linewidth}
                \includegraphics[scale=0.30]{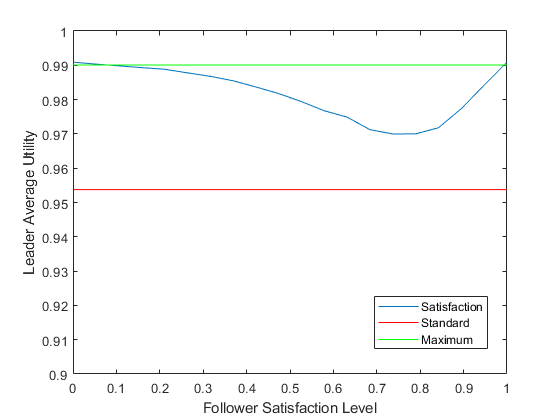}
                \includegraphics[scale=0.30]{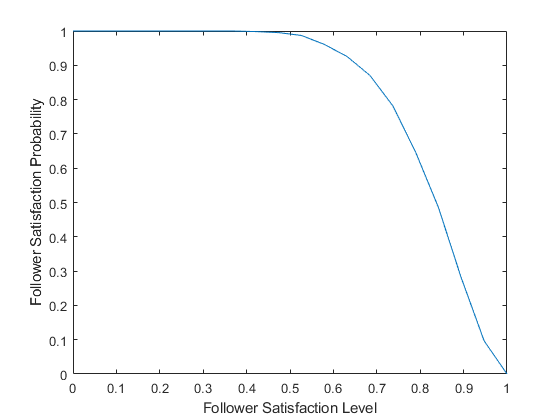}
                \caption{Simulation 2: The maximum possible leader utility of 0.990.\\ Here $N^\ell = 10, N^f = 10$. \label{fig:sim2}}
        \end{subfigure}\\
                \begin{subfigure}[b]{0.5\linewidth}
                \includegraphics[scale=0.30]{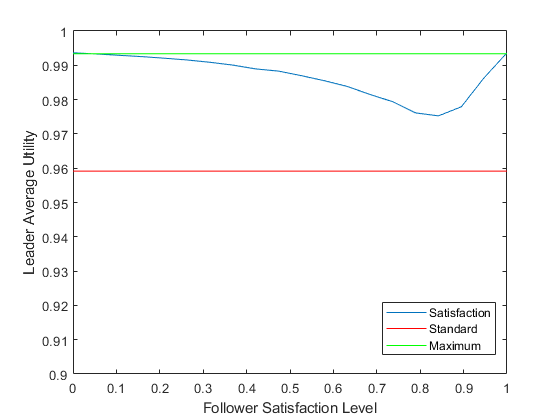}
                \includegraphics[scale=0.30]{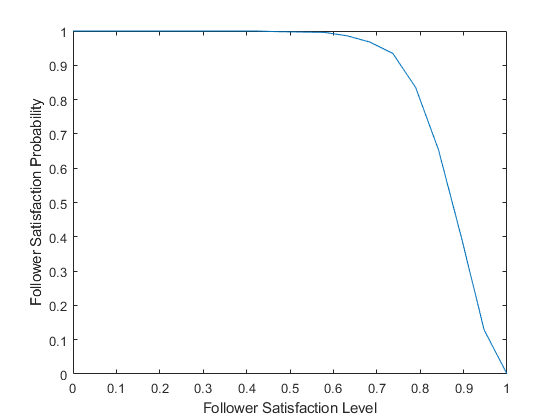}
                \caption{Simulation 3: The maximum possible leader utility of 0.993.\\ Here $N^\ell = 10, N^f = 15$. \label{fig:sim3}}
        \end{subfigure}%
        \begin{subfigure}[b]{0.5\linewidth}
                \includegraphics[scale=0.30]{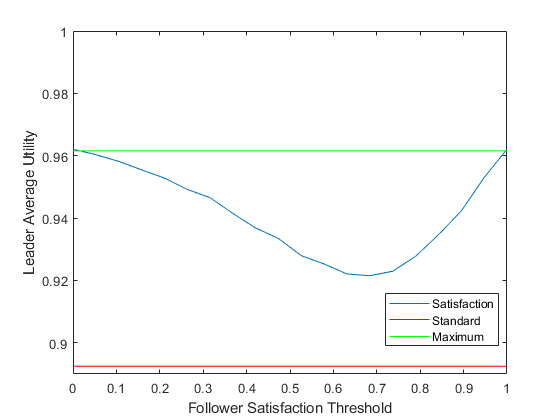}
                \includegraphics[scale=0.30]{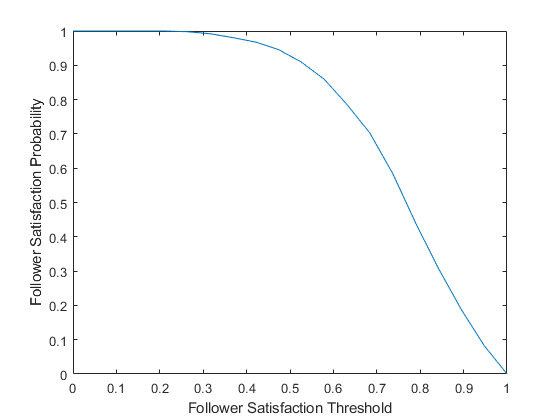}
                \caption{Simulation 4: The maximum possible leader utility of 0.962.\\ Here $N^\ell = 5, N^f = 5$. \label{fig:sim5}}
        \end{subfigure}\\
                \begin{subfigure}[b]{0.5\linewidth}
                \includegraphics[scale=0.30]{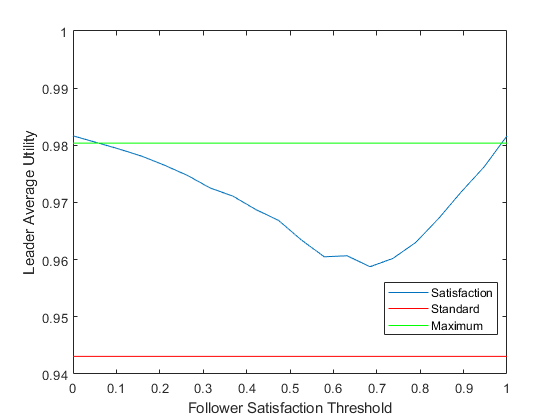}
                \includegraphics[scale=0.30]{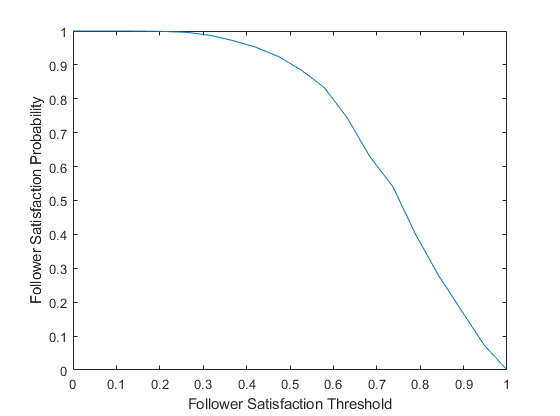}
                \caption{Simulation 5: The maximum possible leader utility of 0.980.\\ Here $N^\ell = 10, N^f = 5$. \label{fig:sim4}}
        \end{subfigure}%
        \begin{subfigure}[b]{0.5\linewidth}
                \includegraphics[scale=0.30]{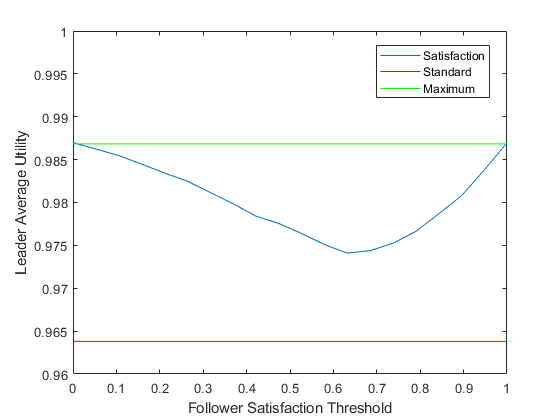}
                \includegraphics[scale=0.30]{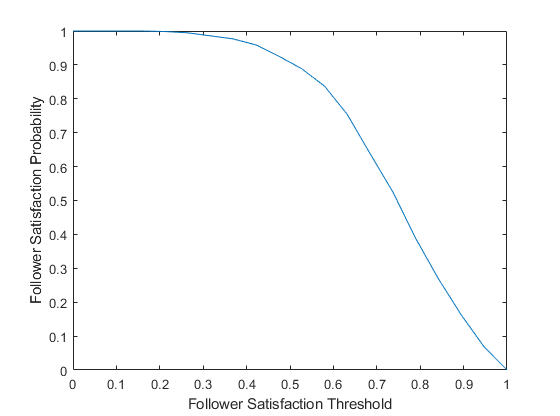}
                \caption{Simulation 6: The maximum possible leader utility of 0.987.\\ Here $N^\ell = 15, N^f = 5$. \label{fig:sim6}}
        \end{subfigure}%
        \caption{The left figure in each subplot shows the leader utilities for the standard Stackelberg duopoly that attained for the satisfaction version. The right figure in each subplot shows the corresponding follower satisfaction probability. The horizontal axis is the follower utility threshold value.}
\end{figure*}

{\it Discussion :}
These observations are made for the specific case of the simulations - utilities are all realisations of i.i.d. uniform $(0,1)$ random variables. These utilities are re-sampled for each realisation of each trial, i.e. the utilities are different for every sample and the multiple trials for each simulation each use different (and independent) utilities. 

\begin{itemize}
\item For a fixed number of leader actions, it appears that as the number of follower actions is increased, the probability of satisfaction increases. 
\item As argued above, for satisfaction threshold $U_-^f$ equal to zero or one, the leader utility attains the maximum average value possible as in this case, the leader simply chooses the pure action pair $(s^\ell,s^f)$ that maximises its utility. These values are shown by the green lines.
\item For a fixed number of leader actions, it appears that as the number of follower actions is increased, the leader utility attained by the standard Stackelberg game increases (red lines) as does the leader utility for the satisfaction version (for all satisfaction thresholds). 
\item It appears that the leader utility as a function of the satisfaction threshold is a convex function with a unique minimum and that this minimum appears to be attained at a higher satisfaction threshold as the number of follower action increases. 
\item The average leader utility for the satisfaction game appears to be always larger than that attained in the standard Stackelberg game with the same utilities. This is a significant result if true (see Proposition 1, sec. \ref{ssec:mult_LP}). 
\item When the number of follower actions is fixed, and the number of leader actions increases, it appears that the follower satisfaction probability decreases slightly whilst the leader utilities for the satisfaction game remain {\it relatively} similar compared to the maximum utility possible.  
\end{itemize}

\subsection{Solution via a Single Linear Program\label{ssec:single_LP}}

In \cite{Con11}, it's shown that a single LP can be used to find the optimal strategy to commit to. More specifically, the LP is, for the standard problem,
\begin{equation*}
\max_{\pi \in \Delta(S^\ell \times S^f)} \ \Er_\pi \ U^\ell(.,.)
\end{equation*}
\begin{align}
& \text{subject\ to\ :} \nonumber \\  
& \forall \, s_f, s_f^\prime \in \Sc^f\ , \sum_{s_\ell \in S^\ell} \ \pi(s_\ell, s_f) \ \lbr U^f(s_\ell,s_f) - U^f(s_\ell, s_f^\prime) \rbr \geq 0 \ ,
\label{eq:1LP1}
\end{align}
together with the usual probability constraints on the (joint) pmf $\pi$. Observe that the inequality constraints in \eqref{eq:1LP1} represent a ``no (conditional) regret'' property for the follower under $\pi$. \footnote{As pointed out in \cite{Con11}, this idea can be extended to multiple followers moving simultaneously, where the followers establish a correlated equilibrium among themselves.} The constraint matrix in \eqref{eq:1LP1} is block diagonal, with each block representing the constraints present in the corresponding (to $s_f$)  one of the LPs \eqref{eq:LP1}. As proven in \cite{Con11}, this LP yields the identical solution to the multiple LP approach, and some benchmarking examples provided in \cite{Con11} are used to argue that, at least using some LP implementations, the single LP approach can be more computationally efficient. The multiple LP approach yields more information (e.g. the best mixed action for the leader given each (pure) follower action) but the single LP formulation allows the ``no regret'' interpretation for the follower. 
 \section{Regret Based Methods\label{sec:regret}}
 
Regret based methods involve an agent reasoning about how its cumulative utility might have been improved if past actions were changed. These ``regrets'' are then used to define the agents' probabilities of play (mixed actions). There are two cases that we are familiar with -- (i) unconditional (or internal) regrets and (ii) conditional (or external) regrets. Unconditional regrets are associated with {\it Hannan equilibria} (HE) whilst conditional regrets are associated with {\it correlated equilibria} (CE).\\

Correlated equilibria are associated with a particular information structure : there is a ``co-ordinator'' which draws the action profile over all players from some probability distribution $\pi$. This distribution is common knowledge. Each agent receives its drawn action from the co-ordinator as private information. Under these equilibrium conditions, there is no incentive for any agent to unilaterally deviate from its recommended action (measured in terms of its expected utility given $\pi$). This idea is known as {\it Bayesian rationality} in game theory (e.g. \cite{Aum87}). We have developed a theory for Bayesian rationality in general satisfaction games and have successfully applied the regret based algorithms to some resource allocation problems in the satisfaction framework \cite{WRN20}.\\ 

In the Stackelberg model, HE or CE may be valid solution concepts, although, in fact, once the assumption of the follower playing BR is made, there is not really any kind of game (it becomes a single agent optimisation problem). However, the joint distribution on play is common knowledge (the leader advertises its strategy to commit to, and the follower plays a corresponding BR), an information setup which is similar to HE/CE. Thus it is of interest to examine so-called ``regret matching'' approaches (e.g. \cite{Hart00,Hart05}) in the context of the Stackelberg duopoly. We comment at this point about how we see these algorithms being applied. We can consider a {\it repeated game} where the game structure remains the same at each stage, and repetition is used to adapt the probabilities of play so that the joint empirical distribution of all players' actions converges to HE/CE as applicable. In the second case, we can use these algorithms {\it in conjunction with a simulator}, to actually compute a HE/CE which can then be used in a ``one-shot'' (or even a repeated) game as an optimal solution.\\ 

The other key point we make is that these regret matching algorithms (which require each agent to have knowledge of all others' actions) can be formulated using ``re-enforcement learning'' concepts. In these approaches, an agent doesn't need any information about the other agents - indeed, an agent is not even ``aware'' it is part of a competitive multi-agent system. If we consider the learning strategies based on \cite{Letch09}), where an ``oracle'' is assumed to be available which can provide follower BR for any leader {\it mixed} strategy, so that the separating hyperplanes defining each BR region in the set of leader mixed strategies can be estimated. Considerable repetition (along with an assumption of stationarity) would be needed in practice (pure leader action drawn repetitively from the same probability distribution) to realise the algorithm of \cite{Letch09}. Regret matching approaches only require knowledge of the {\it realised} actions, so are more applicable in practice. 

\subsection{Regret Matching}
We firstly define the concepts of HE and CE in a general multi-agent system (normal form game). A joint distribution $\pi \in \Delta(\Sc)$ is a Hannan equilibrium (HE) if for each agent $i$, and each action $s^i \in \Sc^i$ it holds that
\begin{align*}
\Er_\pi \lbr U^i(s^i,.) - U^i(.) \rbr \leq 0 \ .
\end{align*}
Under a HE, no agent can increase its expected utility by choosing some fixed action $s^i$ rather than playing a mixed strategy according to its marginal $\pi^i$. A regret matching algorithm for HE \cite{Hart05} is constructed as follows. Each agent $i$ maintains a vector of estimated (unconditional) regrets via
\begin{align*}
\hat{R}^i_t(k) & = \frac{1}{t} \ \sum_{n=1}^t \ U^i \lb k,s_n^{-i} \rb - U^i \lb s_n^i,s_n^{-i} \rb \ ,
\end{align*}
where $s_n$ is the joint action played at stage $n \geq 1$. Here $k$ takes all values in $\Sc^i$. Agent $i$ then constructs its mixed action to play at stage $t+1$ via
\begin{align}
\hat{\pi}^i_{t+1}(k) & = \left\{ \begin{array}{ll} \displaystyle\frac{[\hat{R}^i_t(k)]_+}{\sum_{\ell \in \Sc^i} [\hat{R}^i_t(\ell)]_+} & k \neq s_t^i \\ 1 - \displaystyle\sum_{j \neq k} \frac{[\hat{R}^i_t(j)]_+}{\sum_{\ell \in \Sc^i} [\hat{R}^i_t(\ell)]_+} & k = s_t^i \ , \end{array} \right.
\label{eq:HE_play_probs}
\end{align}
provided the summation in the denominator expressions in \eqref{eq:HE_play_probs} are positive. The notation $[x]_+ = \max(x,0)$, so this condition means that there remains some positive regret terms. If the summation is zero (all regret terms are non-positive), then the play probabilities can be chosen arbitrarily. The main result (\cite{Hart05}, prop. 4) is that the empirical joint distribution over plays converges to a HE. \footnote{The convergence result is about the empirical distribution, not the actual distribution of plays at each stage. In this case, each agent could play its marginal empirical distribution once its regrets became non-positive, which would result (eventually) in all agents playing a HE subsequently.}\\

A joint distribution $\pi \in \Delta(\Sc)$ is a correlated equilibrium (CE) if for each agent $i$, and each pair of actions $j,k \in \Sc^i$ it holds that
\begin{align*}
\Er_{\pi^{-i}|j} \lbr U^i(k,.) - U^i(j,.) \rbr \leq 0 \ ,
\end{align*}
where the expectation is taken over all other agents' actions wrt the conditional distribution $\pi^{-i}|s^i=j$. 
Under a CE, no agent can increase its expected utility by choosing an alternative action $k$ to the recommended action $j$. A regret matching algorithm for CE \cite{Hart00} is constructed as follows. Each agent $i$ maintains a matrix of estimated (conditional) regrets via
\begin{align}
\hat{R}^i_t(j,k) & = \frac{1}{t} \ \sum_{n=1}^t \ \Ibb \lbr s_n^i = j \rbr \ \lb U^i \lb k,s_n^{-i} \rb - U^i \lb j,s_n^{-i} \rb \rb \ .
\label{eq:regrets_CE}
\end{align}
Agent $i$ then constructs its mixed action to play at stage $t+1$ via
\begin{align}
\hat{\pi}^i_{t+1}(k) & = \displaystyle \left\{ \begin{array}{ll} c \, \left[\hat{R}^i_t(s_t^i,k)\right]_+ & k \neq s_t^i\\ \\  1 - c \, \sum_{j \neq k} \left[\hat{R}^i_t(s_t^i,j)\right]_+ & k = s_t^i \ , \end{array} \right.
\label{eq:CE_play_probs}
\end{align}
where the real constant $c > 0$ is chosen so that there is always a positive probability of playing the same action at stage $t+1$ as was played at stage $t$. \footnote{Evidently, $c$ could be different for different agents and/or stages of play.} The main result (\cite{Hart00}, sec. 2) is that the empirical joint distribution over plays converges to a CE.

\subsubsection{Nash Equilibria}

We discuss the relationship between the Stackelberg duopoly and Nash equilibria for the normal form version of the game (simultaneous play). This is important in the context of the convergence of the regret based methods discussed later in this section.\\

The strategic equilibrium concept in games due to Nash is well-known, at least in its pure action form. For a two-player game such as the Stackelberg model but with simultaneous play, a pure joint action $(s_*^\ell,s_*^f)$ is a (pure) Nash equilibrium (NE) if if holds that 
\begin{align}
U^\ell(s^\ell,s_*^f) & \leq U^\ell(s_*^\ell,s_*^f) \ \forall \ s^\ell \in \Sc^\ell  , \text{\ and} \nonumber \\
U^f(s_*^\ell,s^f) & \leq U^f(s_*^\ell,s_*^f) \ \forall \ s^f \in \Sc^f \ .
\label{eq:pure_NE}
\end{align}
We can see from \eqref{eq:pure_NE}, that $s_*^f$ is a follower best-response to the leader pure action $s_*^\ell$. Similarly, $s_*^\ell$ is a leader best-response to the follower pure action $s_*^f$. Thus if an action pair $(s_*^\ell,s_*^f)$ represents the optimal pure strategy to commit to for each player (moving first) and the corresponding BR for the other player (moving second), then $(s_*^\ell,s_*^f)$ is a pure NE of the simultaneous move game \cite{Con06}. However, for pure strategies, committing to a strategy in advance is not necessarily advantageous, whereas for mixed strategies, commitment (and the ability to force BR for the follower) never reduces leader utility \cite{Con06}. \\

More generally \cite{OsRub94}, a (mixed) Nash equilibrium (NE) is a set of probability distributions $\pi^i \in \Delta(\Sc^i)$, one for each agent $i$, such that, for all actions $j,k \in \Sc^i$,
\begin{align}
\pi^i(j) \ \lb \Er \lbr U^i(k,.) \rbr - \Er \lbr U^i(j,.) \rbr \rb \leq 0 \ ,
\label{eq:mixed_NE}
\end{align}
where the expectation is over the other agents actions via the distribution $\prod_{i^\prime \neq i} \pi^{i^\prime}$. We can thus see that NE correspond to points in CE where the joint distribution on play factors into the product of its marginals. This means that, within the CE framework, NE arise under independent actions over agents. From \eqref{eq:mixed_NE}, we see that, for any NE, no positive probability is placed on actions which don't correspond to those maximising the expected utility for that agent (given the NE distributions for the other agents).\\

So in terms of the problem of computing an optimal leader strategy to commit to, generally NE is not a useful optimality concept. This is a problem for the application of some RM algorithms which tend to converge to pure NE when they exist~\cite{NWN23}. We examine this issue in the following sections. 

\subsubsection{Convergence of regret based methods \label{sssec:RM_conv}}

The convergence of the regret matching algorithms described above is in terms of ``approachability'' to the relevant sets (HE or CE). These sets are convex polytopes which contain any pure NE (if they exist) as extreme points. If these algorithms do indeed converge to a point, it is a pure NE \cite{Hart00,Hart05}.  Whilst it is generally true that commitment to a mixed strategy improves the leader's utility, this is not necessarily the case for commitment to pure strategies \cite{Con06}. So from the point of view of the Stackelberg game, NE of the corresponding simultaneous play game may not be useful. This calls into question the applicability of the regret matching ideas (or even more likely, the notion of any equilibria be it HE or CE) if convergence is generally to pure NE. We'll illustrate this with an example in sec. \ref{ssec:examples2}. 

\subsection{Regrets in a Stackelberg Duopoly}

Given the nature of the Stackelberg game (leader commits to a strategy in advance knowing follower plays BR), the problem of determining the optimal leader strategy to commit to is essentially a single agent optimisation problem, provided the follower BR function {\it over leader mixed strategies} is known by the leader. Indeed, the single LP formulation of sec. \ref{ssec:single_LP} confirms this. The structure of this LP indicates that solutions have the ``no positive regrets'' property for the follower. So one possible approach is to use RM for the follower and use the fact that the property of ``universal consistency'' \footnote{Universal consistency means that any player that plays a RM based strategy will have all its regrets becoming non-positive irrespective (within reason) of what the other players do. This was proven for unconditional regrets in \cite{Hart00} (appealing to Blackwell's approachability theorem) and in \cite{NNW17} for conditional regrets.} should guarantee convergence to a ``no-regret'' solution for the follower. However, in the Stackelberg duopoly, the leader's optimal policy to commit to does not necessarily correspond to a ``no-regret'' state \cite{Con11}. Also, we don't generally assume that the leader knows the follower utilities, so it can't directly reason about follower regrets. This is another area for future work. \\

In a general normal form game, RM requires each agent to know the realised actions of all other agents in order to reason about its regrets. In \cite{Hart01}, a RM procedure which does not require an agent to know the other agents' realised actions is presented. This procedure (RM-RL) uses re-enforcement learning ideas to allow an agent to reason about its (conditional) regrets based only on its own past actions and utilities. The application of these methods to the problem at hand is also a topic for future work

\subsection{Examples\label{ssec:examples2}}

In order to demonstrate some behaviours of the RM algorithms, we used an example with i.i.d. uniform $(0,1)$ utilities for the leader and follower (fixed for all examples). There were $N^\ell = 10$ and $N^f = 7$ actions and the realised utilities are listed in appendix A. There are two pure NE for this example : $s^\ell = 9, s^f = 4$, and $s^\ell = 4, s^f = 5$. 

\subsubsection{Unconditional RM}

In the first example, we considered the unconditional (Hannan) RM where the leader has access to the follower BR function (but only for pure leader strategies). Only leader regrets are relevant to this example. The initial leader distribution on actions was uniform. Convergence was determined to occur when the sum of positive regrets (SPR) reduced below $10^{-5}$, then the leader distribution of plays switched to uniform. Recall, that in this RM algorithm, the players may choose any distribution on their plays once convergence to HE is attained. In this example, the optimal leader strategy to commit to is a mixed action $\pi_\ell(2) = 0.8158, \pi_\ell(3) = 0.1842$, with all other values of $\pi_\ell(i), i \notin \{2,3\}$ equal to zero. The corresponding follower BR is action $s^f = 2$. The corresponding utilities are $U^\ell = 0.9681, U^f = 0.7856$. So for this mixed optimal leader strategy to commit to, leader payoff is higher than that attained for either pure NE. \\

Due to randomisation in the selection of actions, different convergence behaviour can be observed for different realisations. Typically, we observe the following :

\begin{itemize}
\item Convergence to either of the two pure NE. Figure \ref{fig:URM1} shows the convergence behaviour of the RM algorithm when convergence is to the NE $(9,4)$. Observe in this case, both regrets go to zero indicating the limit is a  (pure) Hannan NE which corresponds to a pure NE. 
\item Figure \ref{fig:URM2} shows leader and follower SPR when convergence is to the other NE at $(4,5)$. Observe that the follower SPR does not converge to zero even though the limit corresponds to a NE (see comment below).
\item SPR fails to converge to zero for both leader and follower. Typically, in these cases, we get mixed actions for each player with sub-optimal utilities. Results are shown in fig. \ref{fig:URM3}. So the property of universal consistency is not observed, even for the leader.
\end{itemize}

{\it Comment :}
Because we don't use RM to set the play probabilities for the follower (it plays BR), there is no guarantee of convergence of this algorithm. 

\subsubsection{Conditional RM}

We applied the conditional RM algorithm to the case where only leader regrets are used (the follower always plays BR). In this case, over many different realisations, we only observed convergence to one of the pure NE. In all cases, the SPR for both the leader and the follower go to zero. \\

In the case where we applied the standard RM algorithm to both players (no BR function is used), we generally see convergence to the pure NE, however, for some realisations, we see convergence to ``near'' the optimal leader strategy to commit to. Such a situation is shown in fig. \ref{fig:CRM1}. The attained leader utility in this example is 0.9460, and follower utility is 0.7720, only slightly less in each case than the Stackelberg solution. This result gives some indication that (i) optimal leader strategies to commit to (and corresponding follower BR) may be CE (in some cases, although not in general e.g. \cite{Con11}), and as such, (ii) it may be possible to ``engineer'' the conditional RM algorithm to ``encourage'' it to converge to the Stackelberg solution. This is an important area for future research.

\begin{figure*}[!t] 
        \centering
        \begin{subfigure}[b]{0.5\linewidth}
                \includegraphics[scale=0.30]{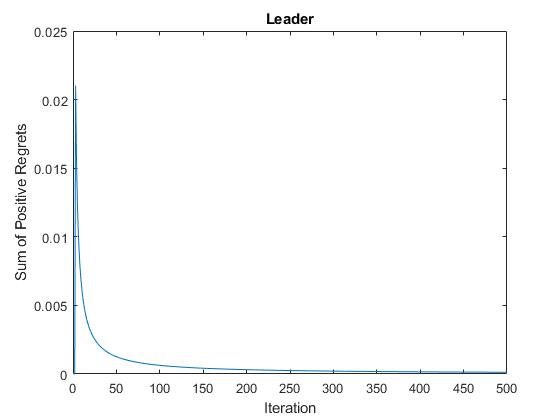}
                \includegraphics[scale=0.30]{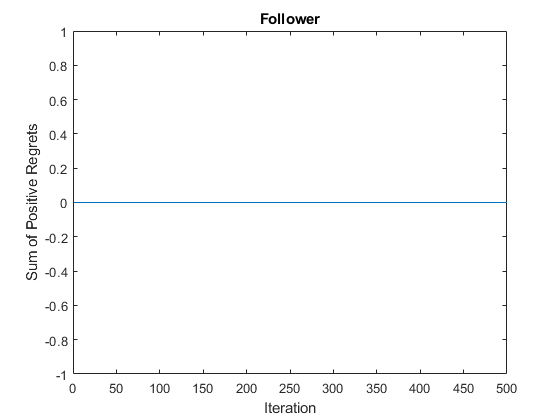}
                \caption{In this case, convergence is to the pure NE at $(9,4)$. \label{fig:URM1}}
        \end{subfigure}%
        \begin{subfigure}[b]{0.5\linewidth}
                \includegraphics[scale=0.30]{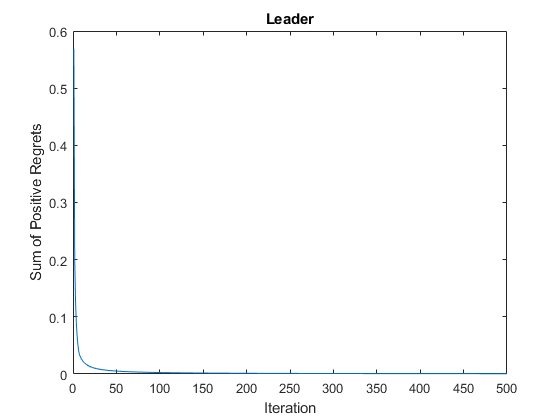}
                \includegraphics[scale=0.30]{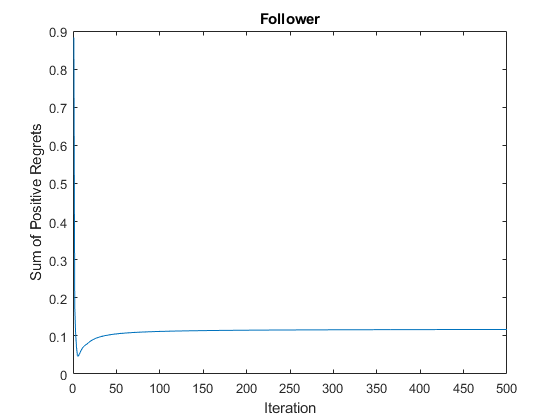}
                \caption[width=0.5\linewidth]{In this case, convergence is to the pure NE at $(4,5)$. \label{fig:URM2}}
        \end{subfigure}%
        \caption{Shows the sum of positive regrets (SPR) for the leader (left) and follower (right).}
\end{figure*}

\begin{figure*}[!t] 
        \centering
        \begin{subfigure}[b]{0.5\linewidth}
                \includegraphics[scale=0.3]{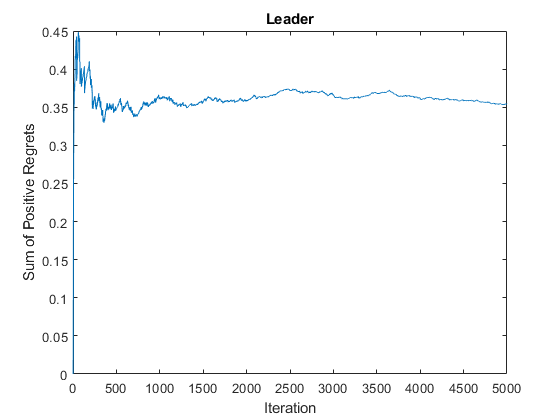}
                \includegraphics[scale=0.3]{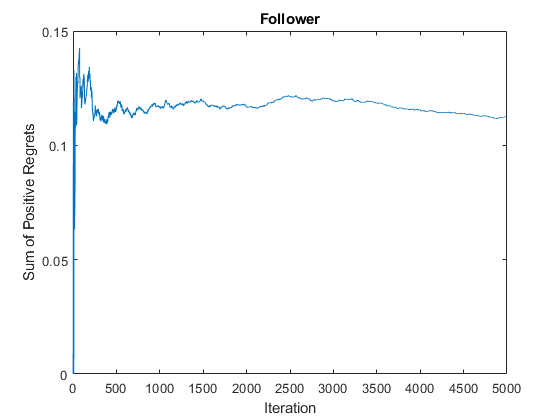}
                \includegraphics[scale=0.3]{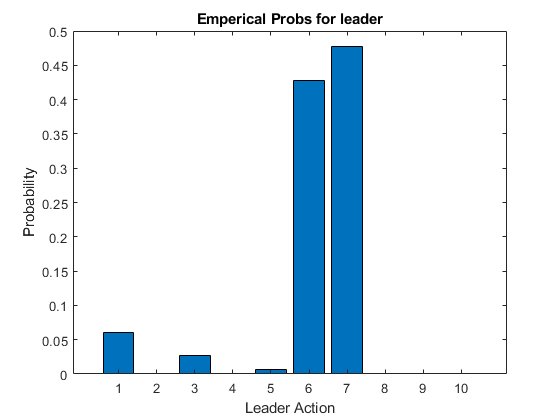}
                \includegraphics[scale=0.3]{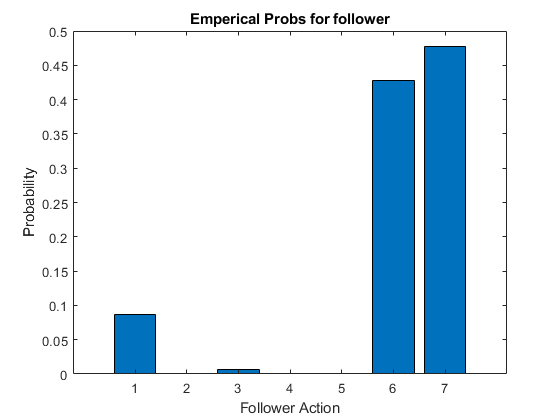}
                \caption{In this case, the algorithm fails to converge to zero of the\\ sum of positive regrets. \label{fig:URM3}}
        \end{subfigure}%
        \begin{subfigure}[b]{0.5\linewidth}
                \includegraphics[scale=0.3]{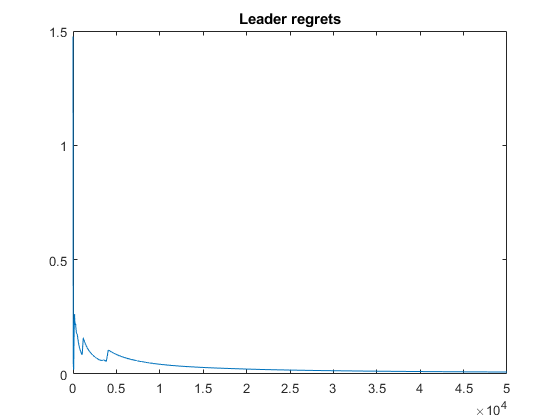}
                \includegraphics[scale=0.3]{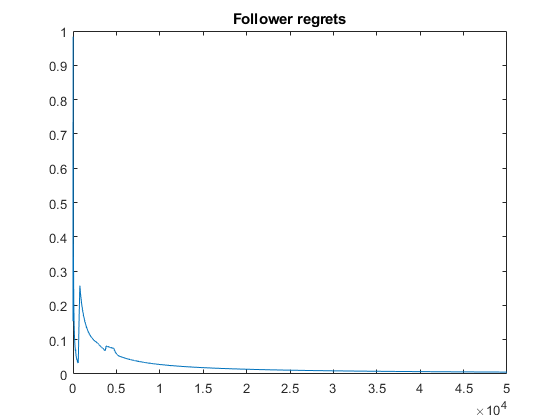}
                \includegraphics[scale=0.3]{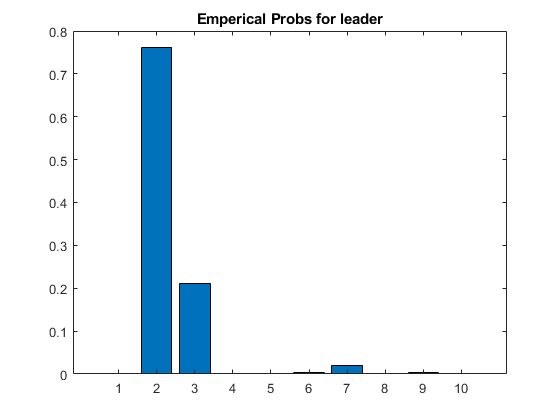}
                \includegraphics[scale=0.3]{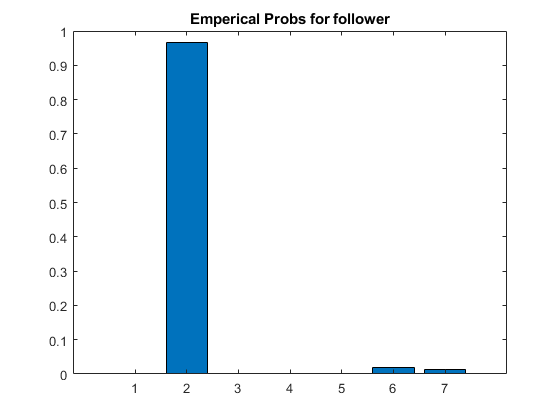}
                \caption[width=0.5\linewidth]{In this case, the algorithm converges towards a CE which appears ``close'' to the optimal leader strategy to commit to (and follower BR). \label{fig:CRM1}}
        \end{subfigure}%
        \caption{Shows the sum of positive regrets (SPR) for the leader (top left in each subplot) and follower (top right in each subplot). The marginal empirical distributions for the leader (bottom left in each subplot) and follower (bottom right in each subplot) are also shown.}
\end{figure*}

\subsection{Regrets in Satisfaction Games}

In \cite{WRN20}, a new paradigm in game theory - Bayesian satisfaction is introduced. This approach combines the idea of Bayesian satisfaction due to Aumann \cite{Aum87} and satisfaction games. Satisfaction equilibria in these games correspond to joint distributions on actions whereby each agent maximises its probability of satisfaction. Examples comparing standard and satisfaction games are included in \cite{WRN20}, where it was shown that, in a radio resource allocation problem, the satisfaction formulation can yield globally more efficient solutions than the normal form game. This is primarily due to agents not claiming resource in excess of their minimum requirement. In \cite{WRN20} it was also shown that RM algorithms can perform significantly better than other algorithms for finding satisfaction equilibria. The consideration of regret based algorithms to Stackelberg satisfaction games remains ongoing work. 
\section{Conclusion and Ongoing Work}

This paper has focused on the study of Stackelberg duopolies within the framework of satisfaction games and on regret-based solution methods for (general) Stackelberg games. The idea of satisfaction in Stackelberg games is introduced and the optimal solution is characterised in terms of probability of satisfaction. For a simple case of i.i.d. uniform $(0,1)$ distributed utilities, we have compared leader optimal utility and follower probability of satisfaction as the satisfaction threshold is varied. We conjecture that if the follower adopts satisfaction, then the optimal (Stackelberg) leader utility is higher. We prove this result for the case of leader commitment to pure strategies.\\

We consider several regret based algorithms for the Stackelberg duopoly. One class of algorithms simply aims to minimise the sum of positive regrets (SPR) for the leader ; the follower plays BR and the leader has access to the BR function {\it for pure leader strategies}. Thus, this is essentially a single agent problem. Simulations suggest that regret matching (RM) algorithms (which are not guaranteed to converge in this case) tend to find pure Nash equilibria for the corresponding simultaneous-move game. These are generally sub-optimal for the leader as a strategy to commit to. However, we did observe convergence to a point ``near'' the Stackelberg solution for conditional RM when regrets were minimised for both players. These observations are likely to be highly problem dependent. \\

%------------------------------------------------------------------------------
\section{Acknowledgements}

The authors acknowledge the support of Defence Science and Technology Group Australia (DST) through the Modelling Complex Warfighting (MCW) strategic research initiative. The author thanks Dr Darryn Reid, Dr Simon Ellis-Steinborner and Dr Glennn Moy from DST for their guidance and insights, and Dr George Stamatescu (The University of Adelaide) for his contributions to the project. 

%---------------------------------------------------------------------------------------------------------
% \newpage

%---------------------------------------------------------------------------------------------------------
\newpage
\appendix
\section{Data}
 
In this appendix, we list the leader and follower utility functions for the examples presented in sec. \ref{ssec:examples2}. This is for reasons of reproducability. In these examples, we used $N^\ell = 10, N^f = 7$. These quantities are i.i.d. realisations of a uniform $(0,1)$ random variable (generated using matlab's \texttt{rand} function). The rows represent leader actions and the columns, follower actions. \\

{\it Leader Utilities :}\\

\begin{tabular}{ccccccc}
    0.8147    & 0.1576    & 0.6557    & 0.7060    & 0.4387    & 0.2760    & 0.7513\\
    0.9058    & 0.9706    & 0.0357    & 0.0318    & 0.3816    & 0.6797     & 0.2551\\
    0.1270    & 0.9572    & 0.8491    & 0.2769    & 0.7655    & 0.6551    & 0.5060\\
    0.9134    & 0.4854    & 0.9340    & 0.0462    & \bcg 0.7952 \ec   & 0.1626    & 0.6991\\
    0.6324    & 0.8003    & 0.6787    & 0.0971   &  0.1869    & 0.1190    & 0.8909\\
    0.0975   &  0.1419    & 0.7577   &  0.8235    & 0.4898   &  0.4984   &  0.9593\\
    0.2785   &  0.4218    & 0.7431   &  0.6948   &  0.4456   &  0.9597   &  0.5472\\
    0.5469    & 0.9157    & 0.3922   &  0.3171    & 0.6463    & 0.3404    & 0.1386\\
    0.9575   &  0.7922    & 0.6555    & \bcb 0.9502 \ec   &  0.7094    & 0.5853    & 0.1493\\
    0.9649    & 0.9595    & 0.1712    & 0.0344    & 0.7547   &  0.2238    & 0.2575\\
\end{tabular}
 
\vspace{5mm}

{\it Follower Utilities :}\\

\begin{tabular}{ccccccc}
    0.8407    & 0.3517    & 0.0759    & 0.1622    & 0.4505    & 0.1067    & 0.4314\\
    0.2543    & 0.8308    & 0.0540    & 0.7943    & 0.0838    & 0.9619    & 0.9106\\
    0.8143    & 0.5853    & 0.5308    & 0.3112    & 0.2290    & 0.0046    & 0.1818\\
    0.2435    & 0.5497    & 0.7792    & 0.5285    & \bcg 0.9133 \ec   & 0.7749    & 0.2638\\
    0.9293    & 0.9172    & 0.9340    & 0.1656   &  0.1524    & 0.8173    & 0.1455\\
    0.3500    & 0.2858    & 0.1299    & 0.6020    & 0.8258    & 0.8687    & 0.1361\\
    0.1966    & 0.7572    & 0.5688    & 0.2630    & 0.5383    & 0.0844    & 0.8693\\
    0.2511    & 0.7537    & 0.4694    & 0.6541    & 0.9961    & 0.3998    & 0.5797\\
    0.6160    & 0.3804    & 0.0119    & \bcb 0.6892 \ec   & 0.0782    & 0.2599    & 0.5499\\
    0.4733    & 0.5678    & 0.3371    & 0.7482    & 0.4427    & 0.8001    & 0.1450\\
\end{tabular}
  
\vspace{5mm}
The pair of actions $s^\ell = 9, s^f = 4$ (marked blue) is a pure NE as is the pair $s^\ell = 4, s^f = 5$ (marked green). There are no mixed NE. 

%-------------------------------------------------------------
\end{document}